\documentclass[3p]{amsart}

\usepackage{latexsym}
\usepackage{amsmath}
\usepackage{epsfig}
\usepackage{amssymb}
\usepackage{enumerate}
\usepackage{color}
\usepackage{expdlist}
\usepackage[british]{babel}

\newtheorem{theorem}{Theorem}
\newtheorem{lemma}{Lemma}

\newcommand{\cT}{{\mathcal T}}

\newcommand{\NN}{{\mathbb N}^+}

\date{\today}

\begin{document}

\raggedbottom

\title[A kernel for computing the hybridization number of multiple trees]{A quadratic kernel for computing the hybridization number of multiple trees}
\author{Leo van Iersel and Simone Linz}
\address{Centrum Wiskunde \& Informatica (CWI), P.O. Box 94079, 1090 GB, Amsterdam, Netherlands.}
\email{l.j.j.v.iersel@gmail.com}
\address{Center for Bioinformatics, University of T\"ubingen, Sand 14, 72076  T\"ubingen, Germany.}
\email{linz@informatik.uni-tuebingen.de}

\begin{abstract}
It has recently been shown that the NP-hard problem of calculating the minimum number of hybridization events that is needed to explain a set of rooted binary phylogenetic trees by means of a hybridization network is fixed-parameter tractable if an instance of the problem consists of precisely two such trees. In this paper, we show that this problem remains fixed-parameter tractable for an arbitrarily large set of rooted binary phylogenetic trees. In particular, we present a quadratic kernel.
\end{abstract}


\maketitle

\section{Introduction}
Phylogenetic trees are a commonly used tool for representing evolutionary relationships. Let~$X$ be a finite set representing for example biological species or, more generally, \emph{taxa}. A rooted \emph{phylogenetic} $X$-\emph{tree} is a rooted tree that has no vertices of outdegree~1 and whose leaves are bijectively labeled by the elements of~$X$. Recently, rooted phylogenetic networks have become increasingly important in analyzing evolutionary histories of sets of taxa whose past may include reticulate evolutionary events such as horizontal gene transfer, hybridization, or recombination. Rooted phylogenetic networks are a generalization of rooted phylogenetic $X$-trees to directed acyclic graphs. In particular, vertices of indegree at least two are called \emph{reticulation vertices} and represent events in which, in the context of hybridization, two distinct ancestral species combine their genomes and form a new species. The \emph{number of reticulations} specified by a reticulation vertex is defined as its indegree minus one while the {\it number of reticulations} specified by a phylogenetic network $N$ is defined as the sum of the number of reticulations over all reticulation vertices in $N$. To quantify the extent to which hybridization events have had an impact on the evolutionary history of a set of present-day species, the following optimization problem has attracted much interest. Let $\cT$ be a set of rooted phylogenetic trees on the same set of taxa. What is the minimum number of reticulations specified by any  phylogenetic network that explains each of the trees in~$\cT$? The decision variant of this problem, called {\sc Hybridization Number}, as well as precise definitions are stated in Section~\ref{sec:prelim}. Since most of the research that is concerned with this question has been done in the context of hybridization, we henceforth refer to a phylogenetic network as a {\it hybridization network} and to reticulations specified by a network as \emph{hybridizations}. 

Since {\sc Hybridization Number} is APX-hard and, thus, NP-hard even for sets of rooted phylogenetic trees consisting of precisely two binary such trees~\cite{bordewich07a}, many theoretical results as well as practical algorithms have been developed for this restricted case. In particular, it has been shown that the two-tree case is fixed-parameter tractable (FPT), regardless of whether the two rooted phylogenetic trees are binary or not~\cite{bordewich07b,linzsemple2009}. Roughly speaking, to establish these results, the authors used several reduction rules that shrink each problem instance to a reduced (weighted) instance whose size is linear in the value of an optimal solution. Subsequent to these results, practical algorithms have been developed that solve {\sc Hybridization Number} for two rooted binary phylogenetic trees~\cite{albrecht2012fast,hybridnet,quantifyingreticulation,whiddenFixed,wuISBRA2010}. Instead of calculating an optimal hybridization network directly, all these algorithms make use of the concept of so-called {\it agreement forests}. Without going into details, an agreement-forest for two rooted binary phylogenetic trees $T$ and $T'$ on the same set of taxa is a collection of disjoint subtrees that are common to $T$ and $T'$. If such a collection is, in a certain sense, acyclic and of minimum size, then its number of elements minus one equates to the solution of {\sc Hybridization Number} for $\cT=\{T,T'\}$~\cite{baroni05}. However, this framework of agreement forests can only be applied to more than two phylogenetic trees if one is solely interested in the minimum number of hybridization \emph{vertices}, but not the actual minimum number of hybridizations specified by any hybridization network that explains the set of trees under consideration. These two numbers are equal in the two-tree case since each hybridization vertex has exactly two parents~\cite{linz2008reticulation}. Given this difficulty and the computational hardness of {\sc Hybridization Number}, it does not come as a surprise that, prior to this paper, there were no exact algorithms that can solve {\sc Hybridization Number} for more than two trees. The only available algorithms are described in~\cite{chen2011algorithms,pirnISMB2010} and, in fact, are heuristics that compute lower and upper bounds for a given instance.

In this paper, we show that {\sc Hybridization Number} remains fixed-parameter tractable if the input to this problem consists of arbitrarily many rooted binary phylogenetic trees on the same set of taxa. This generalization is of significant relevance for applications in (for example) evolutionary biology since biologists usually construct phylogenetic trees for more than two different genes and are interested in the number of hybridizations necessary to explain all reconstructed gene trees simultaneously. Our result shows that, as in the two-tree case, this problem can be solved by using an FPT-algorithm. We hope that this result will facilitate the development of practical algorithms in the same way as it has been the case for the restricted two-tree version of the problem.

The paper is organized as follows. The next section contains some notation and terminology that is used throughout this paper and formally states the decision problem {\sc Hybridization Number}. Section~\ref{sec:fpt} establishes the main result of this paper; thus showing that {\sc Hybridization Number} is fixed-parameter tractable by providing a quadratic kernel. We end this paper with some concluding remarks in Section~\ref{sec:conclu}.

\section{Preliminaries}\label{sec:prelim}

This section provides preliminary definitions that are used throughout this paper and formally states the decision problem {\sc Hybridization Number} for a set of rooted binary phylogenetic trees. Let~$X$ be a finite set. We refer to the elements of~$X$ as \emph{taxa}.

{\bf Phylogenetic trees.} A {\it rooted binary phylogenetic $X$-tree $T$} is a rooted tree whose root has degree two while all other interior vertices have degree three, and whose leaves are bijectively labeled by the elements of~$X$. We identify each leaf with its label and thus refer to~$X$ as the leaf set of~$T$. We regard the edges of~$T$ as being directed away from the root.

{\bf Hybridization networks.} A {\it hybridization network} $N$ on $X$ is a rooted acyclic digraph which has a single root of indegree 0 and outdegree at least 2, has no vertex with indegree and outdegree both 1, and in which the vertices of outdegree 0 are bijectively labeled with the elements of $X$. A vertex whose indegree is at least 2 is called a {\it hybridization vertex}. A hybridization network is {\it binary} if all vertices have indegree and outdegree at most 2 and each hybridization vertex has outdegree  1. Note that a rooted binary phylogenetic $X$-tree is a binary hybridization network on~$X$ with no hybridization vertices.

Let~$N$ be a hybridization network on~$X$. Furthermore, let $X'$ be a subset of~$X$, and let~$T'$ be a rooted phylogenetic $X'$-tree. Then~$T'$ is said to be a \emph{pendant subtree} of~$N$ if it is a subtree that can be detached from~$N$ by deleting a single edge. Furthermore, if $(u,v)$ is an edge of~$N$, we say that $u$ is a parent of $v$ and $v$ is a child of $u$. Note that these definitions hold in particular for rooted phylogenetic trees.

To quantify the number of hybridizations in a hybridization network $N$, the {\it hybridization number} of $N$  is given by $$h(N)=\sum_{v\ne \rho}(d^-(v)-1),$$ where $d^-(v)$ is the indegree of $v$ and $\rho$ is the root of $N$.

Let $N$ again be a hybridization network on $X$, and let $T$ be a rooted binary phylogenetic $X'$-tree, with~$X'\subseteq X$. We say that $T$ is {\it displayed} by $N$ if $T$ can be obtained from $N$ by deleting a subset of the edges and vertices of $N$ and suppressing vertices with indegree and outdegree both 1. In other words,~$N$ displays~$T$ if there exist a subgraph of~$N$ that is a subdivision of~$T$. Intuitively, if $N$ displays $T$, then all of the ancestral relationships of $T$ are visualized by $N$. Furthermore, for a set $\cT$ of rooted binary phylogenetic $X'$-trees, we say that $N$ displays $\cT$ if $N$ displays each tree in $\cT$.

The problem {\sc Hybridization Number} is to compute the minimum hybridization number of a set $\cT$ of rooted binary phylogenetic $X$-trees, which is defined as follows.
$$h(\cT)=\min\{h(N):N\mbox{ is a hybridization network that displays $\cT$}\}.$$

\noindent This problem can formally be stated as the following decision problem.

\noindent{\bf Problem:} {\sc Hybridization Number}\\
\noindent {\bf Instance:} A set $\cT$ of rooted binary phylogenetic $X$-trees and a positive integer $k$.\\
\noindent {\bf Question:} Is $h(\cT) \leq k$?

\noindent In the remainder of this paper, we will exclusively focus our attention on binary hybridization networks. To see that this is sufficient, we need the following lemma~\cite[Lemma 3]{twotrees}.
\begin{lemma}\label{l:binary}
Let $N$ be a hybridization network on $X$ that displays a set of rooted binary phylogenetic $X$-trees. Then there exists a binary hybridization network $N'$ on $X$ that displays $\cT$ such that $h(N')=h(N)$.
\end{lemma}

Let $(\cT,k)$ be an instance of {\sc Hybridization Number}. We will show that two reduction rules described below transform~$(\cT,k)$ into an equivalent instance~$(\cT',k)$ with a quadratic number of taxa. More precisely, $\cT'$ is a collection of rooted binary phylogenetic $X'$-trees such that $h(\cT') \leq k$ if and only if $h(\cT) \leq k$ and~$|X'|\leq 20k^2$.

To describe the reduction rules, we need some additional definitions. Let $\cT$ be a set of rooted binary phylogenetic $X$-trees and let $X'\subseteq X$. A rooted phylogenetic $X'$-tree is a \emph{common pendant subtree} of~$\cT$ if it is a pendant subtree of each element in~$\cT$. Now, let~$T\in\cT$ and let $(x_1,x_2,\ldots ,x_n)$ be a tuple of elements of $X$ with $n\geq 2$, and let~$p_i$ be the parent of the leaf labeled~$x_i$ in~$T$, for each $i\in\{1,2,\ldots ,n\}$. Then, $(x_1,x_2,\ldots ,x_n)$ is called a \emph{chain} of~$T$ if either $(p_n,p_{n-1},\ldots ,p_1)$ is a directed path in~$T$, or  $(p_n,p_{n-1},\ldots ,p_2)$ is a directed path in~$T$ and~$p_1=p_2$. Furthermore, $(x_1,x_2,\ldots ,x_n)$ is a \emph{common chain} of~$\cT$ if it is a chain of each element in~$\cT$.

Let $(\cT,k)$ be an instance of {\sc Hybridization Number}. We are now in a position to state two reduction rules.

\noindent\textbf{Subtree Reduction.} For a common pendant subtree~$T$ of~$\cT$ with at least two leaves, replace, in each element of~$\cT$, the pendant subtree~$T$ by a single leaf labeled by a new taxon (that is not yet in~$X$).
% Replace, in each element of $\cT$, a common pendant subtree of~$\cT$ with at least two leaves by a single leaf with a new label $x$ such that $x\notin X$.

\noindent\textbf{Chain Reduction.} For a  common chain $(x_1,x_2,\ldots ,x_n)$ of $\cT$ with $n>5k$, delete, in each element of $\cT$, the leaves labeled with a member of $\{x_{5k+1},x_{5k+2},\ldots ,x_n\}$ and suppress all vertices with indegree and outdegree both 1.
% If $(x_1,\ldots ,x_n)$ is a common chain of~$\cT$ and~$n>5k$, then delete the leaves labeled $x_{5k+1},\ldots ,x_n$ from each of $T_1,\ldots ,T_t$ and suppress their former parents.

\noindent We remark that similar reductions have been published in the context of calculating the minimum hybridization number as well as the so-called subtree prune and regraft distance for two phylogenies and proven to be important to develop `efficient' algorithms despite the NP-hardness of the underlying problems~\cite{bonet10,bordewich2005computational,bordewich07b}.

To obtain a proof of the kernelization for more than two trees, we need the following notion of {\it generators}. A \emph{binary $k$-reticulation generator} (with~$k\in\NN$) is an acyclic directed multigraph with a single root with indegree~0 and outdegree~1 and all other vertices have indegree~1 and outdegree~2, indegree~2 and outdegree~1, or indegree~2 and outdegree~0. Let~$N$ be a binary hybridization network, with $h(N)=k$, that has no pendant subtrees with two or more leaves. Then, a binary $k$-reticulation generator is said to be the \emph{generator underlying}~$N$ if it can be obtained from~$N$ in the following way. First, delete all leaves of~$N$ and suppress each resulting vertex with indegree and outdegree both~1. Second, if the root has outdegree~2, add a new root with an edge to the old root. For a formal proof showing that the resulting directed multigraph is indeed a binary $k$-reticulation generator, we refer the reader to~\cite[Lemma 4]{kelk2011cycle}. Reversely,~$N$ can be reconstructed from its underlying generator by subdividing edges, adjoining a leaf to each vertex that subdivides an edge, or has indegree 2 and outdegree 0 via a new edge, and deleting the outdegree-1 root. The \emph{sides} of a generator are its edges (the \emph{edge sides}) and its vertices with indegree~2 and outdegree~0 (the \emph{vertex sides}). Thus, each leaf of~$N$ is on a certain side of its underlying generator. To be more formal, let~$x$ be a leaf of~$N$ and let~$p$ be the parent of $x$. If~$p$ is a hybridization vertex, then~$p$ is a vertex side of the underlying generator and we say that~$x$ \emph{is on side}~$p$. If, on the other hand, $p$ has indegree~1 and outdegree~2, then~$p$ is used to subdivide an edge side~$e$ of the underlying generator (because~$N$ has no pendant subtrees with two or more leaves) and we say that~$x$ \emph{is on side}~$e$.

Let $\cT$ be a set of rooted binary phylogenetic $X$-trees with no common pendant subtrees with two or more leaves, and let~$N$ be a binary hybridization network on $X$ that displays $\cT$. Then, clearly,~$N$ has no pendant subtrees with two or more leaves. Let~$G$ be the generator underlying~$N$. A common chain $C=(x_1,x_2,\ldots ,x_n)$ of~$\cT$ is said to \emph{survive} in~$N$ if all elements of $\{x_1,x_2,\ldots,x_n\}$ are on the same edge side of $G$, and $C$ is said to be \emph{atomized} in~$N$ if no two elements of $\{x_1,x_2,\ldots,x_n\}$ are on the same side of $G$.

\noindent\textbf{Kernels and fixed-parameter tractability.} A \emph{kernelization} of a parameterized problem is a polynomial-time algorithm that maps an instance~$x$ with parameter~$k$ to an instance~$x'$ with parameter~$k'$ such that (1)~$(x',k')$ is a yes-instance if and only if~$(x,k)$ is a yes-instance, (2) the size of~$x'$ is bounded by a function~$f$ of~$k$, and (3) the size of~$k'$ is bounded by a function of~$k$. A kernelization is usually referred to as a \emph{kernel} and the function~$f$ as the \emph{size} of the kernel. Thus, a parameterized problem admits a quadratic kernel if there exists a kernelization with~$f$ being a quadratic function. A parameterized problem is \emph{fixed-parameter tractable} if there exists an algorithm that solves the problem in time $O(g(k)|x|^{O(1)})$, with~$g$ being some function of~$k$ and~$|x|$ the size of~$x$. Such an algorithm is called an \emph{FPT-algorithm}. It is well known that a parameterized problem is fixed-parameter tractable if and only if it admits a kernelization and is decidable. However, not for every fixed-parameter tractable problem a kernel of polynomial size is known. Kernels are of particular interest because they can be used as a polynomial-time preprocessing which can be combined with any algorithm solving the problem.

\section{Fixed-parameter tractability of {\sc Hybridization Number}}\label{sec:fpt}
In this section, we establish the following theorem which is the main result of this paper.

\begin{theorem}\label{t:main}
Let $\cT$ be a set of rooted binary phylogenetic $X$-trees, let $\cT'$ be the set of rooted binary phylogenetic $X'$-trees obtained from $\cT$ by applying the subtree reduction as often as possible and subsequently the chain reduction as often as possible, and let~$k\in\NN$. Then, ${h(\cT') \leq k}$ if and only if ${h(\cT) \leq k}$ and~$|X'|\leq 20k^2$. In particular, {\sc Hybridization Number}, parameterized by $k$, is fixed-parameter tractable.
\end{theorem}

\noindent To establish Theorem~\ref{t:main}, we need several lemmas. We start by showing that the subtree reduction does not affect the solution of any instance $(\cT,k)$ of {\sc Hybridization Number}.

\begin{lemma}\label{lem:subtreereduction}
Let $\cT$ be a set of rooted binary phylogenetic $X$-trees and~$k\in\NN$. Furthermore, let $\cT^s$ be the set of trees that results from a single application of the subtree reduction to~$\cT$. Then $h(\cT)\leq k$ if and only if $h(\cT^s)\leq k$.
\end{lemma}

\begin{proof}
First assume that $h(\cT)\leq k$. Then there exists a hybridization network~$N$ that displays $\cT$ such that $h(N)\leq k$. Without loss of generality, choose $N$ such that $h(N)$ is minimized over all hybridization networks that display $\cT$. Consider a common pendant subtree $S$ of $\cT$ that was reduced under an application of the subtree reduction. Then,~$S$ is also a pendant subtree in~$N$ because otherwise there would exist a hybridization network that displays~$\cT$ and has a smaller hybridization number than~$N$. Now, by obtaining a network $N'$ from $N$ by replacing $S$ with a new vertex labeled $s$, it is easily checked that $N'$ is a hybridization network that displays~$\cT^s$. By reversing the argument, it follows that $h(\cT)\leq k$ if and only if $h(\cT^s)\leq k$.
\end{proof}

\noindent As a result of this lemma, we can assume throughout the remainder of this paper that the set of input trees~$\cT$ to {\sc Hybridization Number} has no common pendant subtree. To establish a similar result for the chain reduction (Lemma~\ref{lem:chainreduction}), we need two additional lemmas and some definitions.

For a rooted phylogenetic $X$-tree~$T$ and a subset~$X'$ of~$X$, we define $T|X'$ to be the rooted phylogenetic $X'$-tree obtained from~$T$ by taking the minimal subtree of~$T$ containing all leaves in~$X'$ and suppressing all vertices with indegree~1 and outdegree~1. Given two vertices~$u$ and~$v$ of a hybridization network, we say that~$u$ is an \emph{ancestor} of~$v$ if there is a directed path from~$u$ to~$v$. Furthermore, a vertex of a directed path $P$ is called \emph{internal} if it is not the first or the last vertex of $P$. Lastly, two directed paths~$P_1$ and~$P_2$ are called \emph{internally vertex-disjoint} if there is no vertex of $P_1$ and $P_2$ that is an internal vertex of~$P_1$ and~$P_2$.

\begin{lemma}\label{lem:chains}
Let $\cT$ be a set of rooted binary phylogenetic $X$-trees with no common pendant subtrees with at least two leaves. Then there exists a binary hybridization network~$N$ on $X$ with $h(N)=h(\cT)$ that displays~$\cT$ such that each common chain of~$\cT$ either survives or is atomized in~$N$.
\end{lemma}

\begin{proof}
Let~$N_0$ be a binary hybridization network that displays~$\cT$ such that $h(N_0)=h(\cT)$. Note that such a network exists by Lemma~\ref{l:binary}. We will construct a network~$N$ from $N_0$ that satisfies the statement of the lemma by considering each common chain~$c$ of~$\cT$ that neither survives nor is atomized in~$N_0$ and making changes to the network so that~$c$ survives in~$N$.

Let~$c=(x_1,x_2,\ldots,x_n)$ be a common chain of~$\cT$ that neither survives nor is atomized in $N_0$. Note that~$n\geq 3$, since any chain of two taxa that does not survive is, by definition, atomized, and that~$N_0$ has no pendant subtrees of at least two leaves since~$\cT$ has no common pendant subtrees of at least two leaves.

Let $G_0$ be the generator underlying~$N_0$ and, for convenience, let $C=\{x_1,x_2,\ldots,x_n\}$. Since~$c$ is not atomized, there exist taxa~$x$ and~$x'$ in $C$ that are on the same side~$s$ of~$G_0$. Note that~$s$ can only be an edge side. Let~$p$ and~$p'$ be the parents of~$x$ and~$x'$, respectively, and assume without loss of generality that~$p$ is an ancestor of~$p'$. Let~$e=(g,p)$ be the unique edge entering~$p$. Hence,~$e$ is an edge of the path in~$N_0$ corresponding to side~$s$.

We move all taxa of~$C$ to side~$s$; thereby creating a network in which~$c$ survives. More precisely, we construct networks~$N_1$ and~$N_2$ from~$N_0$ as follows. First, delete all taxa of~$C$ and clean up the resulting network by repeatedly deleting unlabeled outdegree-0 vertices and suppressing vertices with indegree~1 and outdegree~1 until none of these operations is possible (and one has thus obtained a valid hybridization network). Call this intermediate network~$N_1$. We remark that we delete unlabeled outdegree-0 vertices because these arise whenever a leaf is deleted that is on a vertex side of $G_0$. However, by moving such a leaf to an edge side, we reduce the hybridization number of the resulting network which would lead to a contradiction at the end of the proof. Thus, no taxon of~$C$ is on a vertex side of $G_0$. Now, let~$e'$ be the edge of~$N_1$ corresponding to edge~$e$ of~$N_0$. Subdivide~$e'$ by~$n$ vertices~$p_1,p_2,\ldots ,p_n$, creating a directed path $p_n,p_{n-1},\ldots ,p_1$, and introduce a leaf labeled~$x_i$ and an edge~$(p_i,x_i)$ for each~$i\in\{1,2,\ldots ,n\}$. Call the obtained network~$N_2$. 

It remains to show that~$N_2$ displays~$\cT$. Consider any tree~$T\in\cT$. Since~$N_0$ displays~$T$, there exists a subtree~$T_0$ of~$N_0$ that is a subdivision of~$T$. Since~$c$ is a chain of~$T$,~$T_0$ contains a subdivision of a caterpillar on~$C$. In other words, there exist a directed path~$B$ in~$N_0$ and directed paths $L_n,L_{n-1},\ldots ,L_1$ in $N_0$ that start on~$B$ (in that order) and lead to~$x_n,x_{n-1},\ldots ,x_1$ respectively, such that the directed paths $B,L_1,L_2,\ldots,L_n$ are pairwise internally vertex-disjoint. Moreover,~$B$ is chosen such that the first vertex~$r_c$ of~$B$ is the first vertex of~$L_n$ and the last vertex of~$B$ is the first vertex of~$L_1$ and the first vertex of~$L_2$. We next argue that~$p$ is a vertex of~$B$. Since~$x$ and~$x'$ are on the same side of~$G_0$ (and~$p$ is an ancestor of~$p'$), there is a unique directed path from~$p$ to~$p'$. Hence, any path from~$r_c$ to~$p'$ passes through~$p$. Thus,~$B$ passes through~$p$ and it follows that~$p$ is a vertex of~$B$. If edge~$e=(g,p)$ is not an edge of~$B$ (i.e. if $x=x_n$), add~$g$ to~$B$. Now, recall that~$N_1$ was obtained from~$N_0$ by deleting and suppressing vertices. By deleting or suppressing each vertex in~$T_0$ that has been deleted or suppressed in~$N_0$ to obtain~$N_1$, we obtain a subtree~$T_1$ of~$N_1$ that contains a subdivision of $T|(X\setminus C)$. Hence,~$N_1$ displays $T|(X\setminus C)$. Moreover, note that~$e'$ is an edge of $T_1$. Recall that~$N_2$ was obtained from~$N_1$ by subdividing~$e'$ and hanging leaves labeled by elements of~$C$ below the vertices subdividing~$e'$, and observe that~$T$ can be obtained from $T|(X\setminus C)$ by applying the same operations. Therefore, we consider the subtree~$T_2$ of~$N_2$ obtained by applying the same operations to~$T_1$, and conclude that~$T_2$ contains a subdivision of~$T$. It follows that~$N_2$ displays~$T$. Since the above arguments hold for all~$T\in\cT$, it follows that~$N_2$ displays~$\cT$.

By repeating the above construction for each common chain of~$\cT$ that does not survive and is not atomized in $N_0$, we obtain a network~$N$ that displays~$\cT$ such that each common chain of~$\cT$ either survives or is atomized. Moreover, the changes that turned~$N_0$ into~$N$ did not increase the reticulation number. Hence, $h(N)\leq h(N_0)=h(\cT)$. If $h(N) < h(\cT)$, we would obtain a contradiction. Therefore, we conclude that $h(N)=h(\cT)$.
\end{proof}

The following lemma is implicitly in~\cite[Theorem 3.2]{kelk2011cycle}. We include it here for reasons of completeness.

\begin{lemma}\label{lem:generator}
Let~$N$ be a binary hybridization network with $h(N)=k$, and let $G$ be its underlying generator. Then $G$ has at most $4k-1$ edge sides and at most~$k$ vertex sides. In particular, $G$ has at most $5k-1$ sides.
\end{lemma}
\begin{proof}
Let $n_0$ be the number of vertices in $G$ with indegree 2 and outdegree 0, let $n_1$ be the number of vertices in $G$ with indegree 2 and outdegree 1, and let $n_2$ be the number of vertices in $G$ with indegree 1 and outdegree 2. Then, the total indegree of $G$ is $n_2+2n_1+2n_0$ while, considering the root vertex with indegree 0 and outdegree 1, the total outdegree of $G$ is $1+2n_2+n_1$. Hence, by the Handshaking Lemma, we have $n_2+2n_1+2n_0 = 1+2n_2+n_1$ and, therefore, $n_2=n_1+2n_0-1$. Since the number of edge sides of $G$, denoted $|E(G)|$, is equal to the total indegree of $G$ and noting that $n_0+n_1=k$, we have $$|E(G)|=n_2+2n_1+2n_0=3n_1+4n_0-1\leq 4k-1.$$
Furthermore, since each vertex side of $G$ is a vertex with indegree~2, $G$ has at most~$k$ such sides; thereby establishing the lemma.
\end{proof}

\begin{lemma}\label{lem:chainreduction}
Let $\cT$ be a set of rooted binary phylogenetic $X$-trees and~$k\in\NN$. Furthermore, let $\cT^c$ be the set of trees that results from a single application of the chain reduction to~$\cT$. Then $h(\cT)\leq k$ if and only if $h(\cT^c)\leq k$.
\end{lemma}
\begin{proof}
Let $c=(x_1,x_2,\ldots ,x_n)$ be a common chain of~$\cT$ which has been reduced by a chain reduction to a common chain $c'=(x_1,x_2\ldots ,x_{5k})$ of~$\cT^c$. Thus, $n>5k$. 

First, suppose that $h(\cT)\leq k$. Then, by Lemma~\ref{lem:chains}, there exists a binary hybridization network~$N$, with~$h(N)\leq k$, that displays~$\cT$ such that any common chain of~$\cT$ either survives or is atomized in~$N$. Furthermore, by Lemma~\ref{lem:generator}, the generator underlying~$N$ has at most~$5k-1$ sides. Hence, by the pigeonhole principle,~$c$ cannot be atomized in~$N$ and, therefore, survives in~$N$. Now, let~$N'$ be the network obtained from~$N$ by replacing $c$ with $c'$. More precisely, delete all leaves labeled by taxa in $\{x_{5k+1},x_{5k+2},\ldots,x_n\}$ and suppress all resulting vertices of indegree and outdegree both 1. Then,~as $N$ displays $\cT$, it is easily checked that $N'$ displays~$\cT^c$ and~$h(N')\leq k$. Thus, $h(\cT^c)\leq k$.

To show the other direction, suppose that $h(\cT^c)\leq k$. Then, by Lemma~\ref{lem:chains}, there exists a binary hybridization network~$N'$ with~$h(N')\leq k$, that displays~$\cT^c$ such that any common chain of~$\cT^c$ either survives or is atomized in~$N'$. By again using the pigeonhole principle, $c'$ cannot be atomized in $N'$ since it has $5k$ taxa while the generator underlying $N'$ has at most $5k-1$ sides. Hence, $c'$ survives in~$N'$. Now, let~$N$ be the network obtained from~$N'$ by replacing~$c'$ with~$c$. To be precise, let~$e$ be the edge entering the parent, say~$p_{5k}$, of the vertex labeled~$x_{5k}$ in~$N'$. Since $c'$ survives in $N'$, note that $e$ is unique. Subdivide~$e$ by $n-5k$ new vertices $p_{5k+1},p_{5k+2},\ldots ,p_n$, creating a directed path $p_{n},p_{n-1},\ldots ,p_{5k+1}$, and add a leaf labeled $x_i$ and an edge $(p_i,x_i)$ for each $i\in\{5k+1,5k+2,\ldots ,n\}$. 
Then, as $N'$ displays $\cT^c$, it is easily checked that $N$ displays~$\cT$ and has~$h(N)\leq k$. Thus, $h(\cT)\leq k$.
\end{proof}

We next show that the subtree and chain reduction can be applied to a collection of rooted binary phylogenetic $X$-trees until the label set of the resulting collection of trees has size bounded by a quadratic function of $h(\cT)$. For the proof, we follow an approach similar to the one taken by Kelk et al.~\cite[Lemma 3.2]{kelk2011cycle}.

\begin{lemma}\label{lem:kernelsize}
Let $\cT$ be a set of rooted binary phylogenetic $X$-trees, let $\cT'$ be the set of rooted binary phylogenetic $X'$-trees obtained from $\cT$ by applying the subtree and chain reduction until no further reduction is possible, and let~$k\in\NN$. If $h(\cT)\leq k$, then $|X'|\leq 20k^2$.
\end{lemma}
\begin{proof}
As $h(\cT)\leq k$, it follows from Lemmas~\ref{lem:subtreereduction} and~\ref{lem:chainreduction} that $h(\cT')\leq k$. Let $N$ be a binary hybridization network that displays~$\cT'$ such that~$h(N)\leq k$. Furthermore, let~$G$ be its underlying binary $h(N)$-reticulation generator.

Observe that~$N$ has no pendant subtrees of size at least 2, since otherwise~$\cT'$ would have a common pendant subtree; thereby contradicting that the subtree reduction has been applied as often as possible. Furthermore,~$N$ does not have more than~$5k$ leaves that are on the same side of~$G$, since otherwise~$\cT'$ would have a common chain of size greater than~$5k$, thereby contradicting that the chain reduction has been applied as often as possible.

Thus,~$N$ has one leaf per vertex side of~$G$ and at most~$5k$ leaves per edge side of~$G$. By Lemma~\ref{lem:generator} (and because $h(N)\leq k$),~$G$ has at most~$4k-1$ edge sides and at most~$k$ vertex sides. Thus, the total number of leaves is at most $5k\cdot (4k-1) + k = 20k^2 - 4k < 20k^2$. It now follows that $|X'|\leq 20k^2$.
\end{proof}

Now, Theorem~\ref{t:main} follows from Lemmas~\ref{lem:subtreereduction},~\ref{lem:chainreduction} and~\ref{lem:kernelsize}.

\section{Concluding remarks}\label{sec:conclu}

While Theorem~\ref{t:main} proves the existence of an FPT-algorithm to solve {\sc Hybridization Number}, it does not describe an explicit algorithm to do so. In order to obtain such an algorithm, one needs an exponential-time exact algorithm to solve an instance of {\sc Hybridization Number} after it has been kernelized. One possible way to design an FPT-algorithm for {\sc Hybridization Number} is the following. Theorem~2 of~\cite{elusiveness} establishes an algorithm---called {\sc Clustistic}---that, given a set of rooted binary phylogenetic trees and an integer~$k$, finds all binary hybridization networks that represent all clusters of the trees (in the so-called {\it softwired sense}, see e.g.~\cite{twotrees}) and have hybridization number at most~$k$. Since any network that displays a given set of rooted binary phylogenetic trees also represents all clusters of those trees, {\sc Clustistic} finds it. Thus, an exponential-time exact algorithm for {\sc Hybridization Number} can be obtained by using {\sc Clustistic} and checking for each returned network if it displays the input trees (e.g. using the algorithm in~\cite{Kanj2008}, which is exponential in the number of hybridizations of a given hybridization network). In combination with the presented kernelization, this leads to an FPT-algorithm for {\sc Hybridization Number}. We omit the details of this algorithm as its theoretical worst-case running time is not necessarily the best and we expect that methods are possible that are much faster in practice. 
We also remark that if one allows weighted chains, as in \cite{bordewich07b}, then a slightly modified chain reduction can be used to obtain a linear kernel for a modified problem, where each common chain is associated with a weight.

A major open problem is to show whether or not it is also fixed-parameter tractable to compute the minimum hybridization number of a set of arbitrary rooted phylogenetic trees; thus allowing for trees that are nonbinary.

\noindent{\it Acknowledgements.} We thank Steven Kelk for many helpful discussions on the topic of this paper. Leo van Iersel was funded by a Veni grant of The Netherlands Organization for Scientific Research (NWO) and Simone Linz by the University of T\"ubingen and a travel fellowship by the German Academic Exchange Service (DAAD).

\bibliographystyle{plain}
\bibliography{MultiTreeFPT}

\begin{thebibliography}{10}

\bibitem{albrecht2012fast}
B.~Albrecht, C.~Scornavacca, A.~Cenci, and D.H. Huson.
\newblock Fast computation of minimum hybridization networks.
\newblock {\em Bioinformatics}, 28(2):191--197, 2012.

\bibitem{baroni05}
M.~Baroni, S.~Gr\"unewald, V.~Moulton, and C.~Semple.
\newblock Bounding the number of hybridisation events for a consistent
  evolutionary history.
\newblock {\em J. of Math. Biol.}, 51:171--182, 2005.

\bibitem{bonet10}
M.L. Bonet and K.~St. John.
\newblock On the complexity of u{SPR} distance.
\newblock {\em IEEE/ACM Trans. Comput. Biol. Bioinf.}, 7(3):572--576, 2010.

\bibitem{bordewich2005computational}
M.~Bordewich and C.~Semple.
\newblock On the computational complexity of the rooted subtree prune and
  regraft distance.
\newblock {\em Ann. Comb.}, 8(4):409--423, 2005.

\bibitem{bordewich07b}
M.~Bordewich and C.~Semple.
\newblock Computing the hybridization number of two phylogenetic trees is
  fixed-parameter tractable.
\newblock {\em IEEE/ACM Transactions on Computational Biology and
  Bioinformatics}, 4:458--466, 2007.

\bibitem{bordewich07a}
M.~Bordewich and C.~Semple.
\newblock Computing the minimum number of hybridization events for a consistent
  evolutionary history.
\newblock {\em Discrete Appl. Math.}, 155(8):914--928, 2007.

\bibitem{hybridnet}
Z.Z. Chen and L.~Wang.
\newblock Hybrid{NET}: a tool for constructing hybridization networks.
\newblock {\em Bioinformatics}, 26(22):2912--2913, 2010.

\bibitem{chen2011algorithms}
Z.Z. Chen and L.~Wang.
\newblock Algorithms for reticulate networks of multiple phylogenetic trees.
\newblock {\em IEEE/ACM Trans. Comput. Biol. Bioinf.}, 9(2):1545--5963, 2012.

\bibitem{quantifyingreticulation}
J.~Collins, S.~Linz, and C.~Semple.
\newblock Quantifying hybridization in realistic time.
\newblock {\em J. Comput. Biol.}, 18:1305--1318, 2011.

\bibitem{twotrees}
L.J.J.~van Iersel and S.M. Kelk.
\newblock When two trees go to war.
\newblock {\em J. Theo. Biol.}, 269(1):245--255, 2011.

\bibitem{Kanj2008}
I.A. Kanj, L.~Nakhleh, C.~Than, and G.~Xia.
\newblock Seeing the trees and their branches in the network is hard.
\newblock {\em Theo. Comp. Science}, 401:153--164, 2008.

\bibitem{kelk2011cycle}
S.~Kelk, L.~van Iersel, N.~Lekic, S.~Linz, C.~Scornavacca, and L.~Stougie.
\newblock Cycle killer... qu'est-ce que c'est? on the comparative
  approximability of hybridization number and directed feedback vertex set.
\newblock {\em arXiv:1112.5359v1 [math.CO]}, 2011.

\bibitem{elusiveness}
S.M. Kelk, C.~Scornavacca, and L.J.J. van Iersel.
\newblock On the elusiveness of clusters.
\newblock {\em IEEE/ACM Trans. Comput. Biol. Bioinf.}, 9(2):517--534, 2012.

\bibitem{linz2008reticulation}
S.~Linz.
\newblock Reticulation in evolution.
\newblock {\em PhD thesis, Heinrich-Heine Universit{\"a}t, D{\"u}sseldorf,
  Germany}, 2008.

\bibitem{linzsemple2009}
S.~Linz and C.~Semple.
\newblock Hybridization in nonbinary trees.
\newblock {\em IEEE/ACM Trans. Comput. Biol. Bioinf.}, 6(1):30--45, 2009.

\bibitem{whiddenFixed}
C.~Whidden, R.G. Beiko, and N.~Zeh.
\newblock Fixed-parameter and approximation algorithms for maximum agreement
  forests, 2011.
\newblock arXiv:1108.2664v1 [q-bio.PE].

\bibitem{pirnISMB2010}
Y.~Wu.
\newblock Close lower and upper bounds for the minimum reticulate network of
  multiple phylogenetic trees.
\newblock {\em Bioinformatics}, 26:i140--i148, 2010.
\newblock Special issue: Proceedings of Intelligent Systems for Molecular
  Biology 2010 (ISMB2010), 10th-13th September 2010, Boston USA.

\bibitem{wuISBRA2010}
Y.~Wu and W.~Jiayin.
\newblock Fast computation of the exact hybridization number of two
  phylogenetic trees.
\newblock In {\em Bioinformatics Research and Applications (ISBRA)}, volume
  6053, pages 203--214, 2010.

\end{thebibliography}

\end{document}